\documentclass{article}

\usepackage{arxiv}

\usepackage[utf8]{inputenc} % allow utf-8 input
\usepackage[T1]{fontenc}    % use 8-bit T1 fonts
\usepackage{hyperref}       % hyperlinks
\usepackage{url}            % simple URL typesetting
\usepackage{booktabs}       % professional-quality tables
\usepackage{amsfonts}       % blackboard math symbols
\usepackage{nicefrac}       % compact symbols for 1/2, etc.
\usepackage{microtype}      % microtypography
\usepackage{lipsum}		% Can be removed after putting your text content
\usepackage{graphicx}
\usepackage{doi}
\usepackage{xcolor}
\usepackage{amsmath}
\usepackage{amsfonts}
\usepackage{graphicx}
\usepackage{mathtools}
\usepackage{stmaryrd}
\usepackage[ruled,vlined]{algorithm2e}
\usepackage{tikz}
\usepackage{svg}
\usepackage{amsmath}
\usepackage{amsthm}

\usepackage{amssymb}
\usepackage{mathabx}
\usepackage{hyperref}
\usepackage{lipsum}
\usepackage{enumerate}
\usepackage{booktabs}
\hypersetup{
    colorlinks=true,
    linkcolor=black,
    urlcolor=blue,
    pdftitle={},
    pdfpagemode=FullScreen,
    }

\usetikzlibrary{shapes.geometric, arrows}
\usetikzlibrary{calc,fit,arrows}
\allowdisplaybreaks
\newtheorem{theorem}{Theorem}[section]

\newtheorem{assumption}[theorem]{Assumption}

\newtheorem{proposition}[theorem]{Proposition}
\newtheorem{Definition}[theorem]{Definition}

\newtheorem{remark}[theorem]{Remark}

\newcommand{\R}{{\mathbb{R}}}
\newcommand{\N}{{\mathbb{N}}}

\title{Controller for Incremental Input-to-State Practical Stabilization of Partially Unknown Systems with Invariance Guarantees}

\author{P Sangeerth$^*$, David Smith Sundarsingh$^*$, Bhabani Shankar Dey, and Pushpak Jagtap % <-this % stops a space
\thanks{This work was supported in part by the ARTPARK, and the Siemens.}% <-this % stops a space
\thanks{$^*$Authors contributed equally.}% <-this % stops a space
\thanks{D.S.Sundarsingh is with Washington University in St. Louis, Missouri, USA {\tt\small d.s.sundarsingh@wustl.edu}} \thanks{P. Sangeerth, B.S. Dey, and P. Jagtap are with Robert Bosch Centre for Cyber-Physical Systems, IISc, Bangalore, India {\tt\small sangeerthp,bhabanishan1,pushpak@iisc.ac.in}}}
\usepackage{booktabs,siunitx}
\usepackage{tabularx,ragged2e}
\newcolumntype{L}{>{\RaggedRight\hangafter=1\hangindent=1em}X}
\renewcommand{\shorttitle}{\textit{arXiv} Template}

\hypersetup{
pdftitle={A template for the arxiv style},
pdfsubject={q-bio.NC, q-bio.QM},
pdfauthor={David S.~Hippocampus, Elias D.~Striatum},
pdfkeywords={First keyword, Second keyword, More},
}

\begin{document}

\maketitle
\thispagestyle{empty}
\pagestyle{empty}

%%%%%%%%%%%%%%%%%%%%%%%%%%%%%%%%%%%%%%%%%%%%%%%%%%%%%%%%%%%%%%%%%%%%%%%%%%%%%%%%
\begin{abstract}
Incremental stability is a property of dynamical systems that ensures the convergence of trajectories with respect to each other rather than a fixed equilibrium point or a fixed trajectory. In this paper, we introduce a related stability notion called incremental input-to-state practical stability ($\delta$-ISpS), ensuring safety guarantees. We also present a feedback linearization based control design scheme that renders a partially unknown system incrementally input-to-state practically stable and safe with formal guarantees. To deal with the unknown dynamics, we utilize Gaussian process regression to approximate the model. Finally, we implement the controller synthesized by the proposed scheme on a manipulator example.
\end{abstract}
%%%%%%%%%%%%%%%%%%%%%%%%%%%%%%%%%%%%%%%%%%%%%%%%%%%%%%%%%%%%%%%%%%%%%%%%%%%%%%%%
\section{Introduction}
Incremental stability is a strong property of nonlinear systems that focuses on the convergence of trajectories with respect to each other rather than to a specific trajectory or equilibrium point. In recent years, this notion has gained significant attention due to its application in synchronization of complex networks \cite{synchComplex} and interconnected systems \cite{synchOsci, InterIncre}, and the construction of scalable symbolic models for nonlinear control systems \cite{girard2014approximately,jagtap2020symbolic}.

A particular class of incremental stability, Incremental Input-to-State Stability ($\delta$-ISS), has been extensively studied. Similar to other stability notions, Lyapunov functions have been used to characterize $\delta$-ISS as shown in \cite{angeli,zamanicharacterize}. Furthermore, state feedback controllers have been designed to render a class of control systems $\delta$-ISS. A few examples include works on stochastic systems \cite{pushpakHamilton}, unstable non-smooth control systems \cite{zamaninonsmooth}, and backstepping approaches \cite{backsteppingzamani}. However, all approaches require knowledge of the system in order to design a controller. To the best of the author's knowledge, there has been no work on controller synthesis for $\delta$-ISS stabilization of an unknown system, along with giving formal safe guarantees. In this paper, we design a controller for strict feedback nonlinear systems with partially unknown dynamics by using the Gaussian process (GP) to approximate the model of the system. This will act as the solution for a large class of systems, including Euler-Lagrange systems.

Gaussian process \cite{GPBook} is a data-driven and non-parametric learning approach that has been used in system identification for control. Several works exist in the literature where GP has been used in this fashion, including works on tracking control \cite{tracking}, feedback linearization \cite{feedbackLinearization}, control Lyapunov function (CLF) approach \cite{hircheControl}, and control barrier functions (CBF) \cite{JagtapGP}. Unfortunately, GP models are not perfect, making it infeasible to design a controller that ensures strict $\delta$-ISS by using the learned model of the system. To solve this problem, motivated by the notion of  Input-to-State practical Stability (ISpS) \cite{ISpS}, we first relax the notion of $\delta$-ISS property to Incremental Input-to-State Practical Stability ($\delta$-ISpS). As a GP-based method is a data-driven technique, there is no state space invariance guarantee unless taken care of explicitly. Therefore, we provide a characterization of the $\delta$-ISpS property in terms of $\delta$-ISpS Lyapunov functions and then develop a feedback linearization control scheme that ensures $\delta-$ISpS with state-space invariance safety guarantees using a CBF approach \cite{kolathaya2018input,jagtap2020control}. In our work, we restrict safety guarantees to be state-space invariant.

To the best of our knowledge, this work is the first to design a controller that guarantees incremental stabilization ($\delta$-ISpS) and safety simultaneously for a partially unknown system on a compact state-space. It also provides the first definition and Lyapunov-based characterization for $\delta$-ISpS enforcing safety. We address the synthesis problem by first learning the unknown dynamics using a GP \cite{HircheGP} with probabilistic accuracy guarantees. We then provide a feedback controller design and a corresponding $\delta$-ISpS Lyapunov function that renders the system $\delta$-ISpS. Finally, we demonstrate the trajectory convergence of the synthesized controller in a case study where the dynamics is partially known.

\textbf{Notations:}
The sets of real, positive real, non-negative real, and positive integers are $\R$, $\R^+$, $\R^+_0$, and $\N$, respectively. $\R^n$ denotes an $n$-dimensional Euclidean space, and $\R^{n\times m}$ is the space of $n\times m$ real-valued matrices. For $x\in \mathbb{R}^n$, $\lVert x \rVert$ is the Euclidean norm. For a measurable function $\upsilon:\R^+_0\rightarrow\R^n$, $\lVert\upsilon\rVert_{\infty}:=(ess)\sup\left\{\lVert \upsilon(t)\rVert,t\geq 0\right\}$ is the (essential) supremum. $\mathcal{N}(\mu,\rho)$ denotes the multivariate normal distribution, with mean $\mu\in\R^n$ and covariance $\rho \in\R^{n\times n}$. The reproducing kernel Hilbert space (RKHS) is a Hilbert space of square-integrable functions. The RKHS norm is denoted by $\lVert f\rVert_{k}$, where $f$ is a function, $k:X\times X\rightarrow\R^+_0$ is a kernel (symmetric positive definite), and $X\subset\R^n$. The RKHS includes functions of the form $f(x)=\sum_i a_ik(x,x_i)$ ($a_i\in\R$, $x,x_i\in X$). (See \cite{InfoGain} for details). A continuous function $\alpha:\R^+_0\rightarrow\R^+_0$ is class-$\mathcal{K}$ if it is strictly increasing and $\alpha(0)=0$. If $\alpha\in\mathcal{K}$ and is radially unbounded ($\alpha(r)\rightarrow\infty$ as $r\rightarrow\infty$), it is a class-$\mathcal{K}_{\infty}$ function. A continuous function $\beta:\R^+_0\times\R^+_0\rightarrow\R^+_0$ belongs to class-$\mathcal{KL}$ if, for fixed $s$, $\beta(\cdot,s)\in\mathcal{K}_{\infty}$, and for fixed $r$, $\beta(r,\cdot)$ is decreasing, with $\beta(r,s)\rightarrow 0$ as $s\rightarrow\infty$. $\mathcal{I}_d$ represents the identity function.
%%%%%%%%%%%%%%%%%%%%%%%%%%%%%%%%%%%%%%%%%%%%%%%%%%%%%%%%%%%%%%%%%%%%
\section{Incremental Input to State practical Stability}\label{section:prelim}
%\subsection{Control System}
We consider a continuous-time control system $\Sigma$ represented as
\begin{align}
    \Sigma: \dot x=f(x,u),
\end{align}
where $x(t)\in X \subset \R^d$ represents the state of the system in the state-space $X$ which is compact, $u(t) \in U \subseteq \mathbb{R}^m$ represents the input to the system from the input space $U$. The map $f:X\times U\rightarrow \mathbb{R}^d$ is assumed to satisfy the local Lipschitz continuity assumption to ensure the existence and uniqueness of trajectories \cite{Khalil:1173048}. We now consider the application of a feedback controller $p:X \times W \rightarrow U$, resulting in a closed loop continuous time control system, \begin{gather}\label{eqn:sys_del_isps}
       \Sigma: \dot x=f(x,p(x,\upsilon)),
    \end{gather}
where $\upsilon(t) \in W \subset \mathbb{R}^m$ represents the external input from the external input set $W$.  We represent the trajectory of the closed loop system by $x_{a\upsilon}:\R^+_0 \rightarrow X$, under the input signal $\upsilon:\mathbb{R}^+_0 \rightarrow W$ starting from the initial condition $x_{a \upsilon}(0):=a$. 

Now, we introduce the notion of Incremental Input-to-State practical Stability ($\delta$-ISpS), inspired by \cite{ISpS}, and its characterization using a $\delta$-ISpS Lyapunov Function.

\begin{Definition}[$\delta$-ISpS]\label{def:deltaISpS}
A control system $\Sigma$ from \eqref{eqn:sys_del_isps} is \textit{incrementally input-to-state practically stable} if there exist functions $\beta\in\mathcal{KL}$, $\gamma\in\mathcal{K}_{\infty}$ and a constant $c>0$ such that for any $t\in\R^+_0$, any external inputs $\upsilon,\upsilon':\mathbb{R}^+_0 \rightarrow W$, and any initial states $a, a'\in X$, the inequality:
\begin{align}\label{eqn:deltaISpS}
 	\lVert x_{a\upsilon}(t)\hspace{-0.25em}-\hspace{-0.25em}x_{a'\upsilon'}(t)\rVert \hspace{-0.25em}\leq\hspace{-0.25em} \beta(\lVert a\hspace{-0.25em}-\hspace{-0.25em}a'\rVert,t)\hspace{-0.25em} +\hspace{-0.25em} \gamma(\lVert \upsilon\hspace{-0.25em}-\hspace{-0.25em}\upsilon'\rVert_{\infty}) \hspace{-0.25em}+\hspace{-0.25em} c.
\end{align}
\end{Definition}
holds true. Note that if $c=0$, the system is \textit{incrementally input-to-state stable} \cite{angeli}. In this paper, we aim to present a data-driven approach to tackle the problem of incremental stability. This requires working with compact sets, so we first introduce the notion of $\delta$-ISpS-CLF for compact sets. To do this, we introduce the notion of forward invariance.
\begin{Definition}[Robustly Forward Invariant Set \cite{liu2019compositional}]
    A set $\mathcal{C}$ is said to be robustly forward invariant with respect to the system \eqref{eqn:sys_del_isps} if for every $a \in \mathcal{C}$ and for all $\upsilon:\mathbb{R}^+_0 \rightarrow W$, there exists some control input $u(t):=p(x_{a\upsilon}(t),\upsilon(t)) \in U$ such that $x_{a\upsilon}(t) \in \mathcal{C}$, for all $t \geq 0$. 
    The controller $p$, which makes the system robustly forward invariant, is called a forward invariant controller.
\end{Definition}
Now, we introduce the notion of $\delta$-ISS-CLF for the closed-loop system given in \eqref{eqn:sys_del_isps} where the sets $X\subset \mathbb{R}^n$ and $W \subset \mathbb{R}^m$ are compact and $X$ is considered to be robustly forward invariant under the controller $p$. 
\begin{Definition}[$\delta$-ISpS Control Lyapunov Function]\label{def:deltaISpSLyapunov}
A differentiable function $V:X \times X\rightarrow\R^+_0$ is a $\delta$-ISpS-CLF for the closed-loop system \eqref{eqn:sys_del_isps}, if there exist a forward invariant controller $p:X \times W \rightarrow U$, functions $\overline{\alpha},\underline{\alpha},\sigma\in\mathcal{K}_{\infty}$, and constants $\Tilde{c},\kappa \in\R^+$, such that:
\begin{enumerate}[(i)]
 	\item $\forall x,x'\in X$, $\underline{\alpha}(\lVert x-x'\rVert)\leq V(x,x')\leq\overline{\alpha}(\lVert x-x'\rVert)$;\label{item:lyapunov1}
 	\item $\forall x,x'\in X$ and $\forall \upsilon,\upsilon':\mathbb{R}^+_0 \rightarrow W$,
        $\dot V(x,x')\leq-\kappa V(x,x')+\sigma(\lVert \upsilon-\upsilon'\rVert_{\infty})+\Tilde{c}.$\label{item:lyapunov2}
\end{enumerate}
\end{Definition}
The following theorem describes $\delta$-ISpS property of a system in terms of the existence of a $\delta$-ISpS-CLF. 
\begin{theorem}\label{theorem:ISpSLyapunov}
    The closed-loop system $\Sigma$ is $\delta$-ISpS within the state space $X$ with respect to the external input $\upsilon$, if it admits a  $\delta$-ISpS-CLF as defined in Definition \ref{def:deltaISpSLyapunov}.
\end{theorem}
\begin{proof}
    The proof is inspired by the proof of \cite[Theorem 2.6]{deltaISS}
    Consider a Lyapunov function satisfying the conditions \eqref{item:lyapunov1}-\eqref{item:lyapunov2} in Definition \ref{def:deltaISpSLyapunov}. 
    By using property \eqref{item:lyapunov1} in Definition \ref{def:deltaISpSLyapunov}, we obtain,
    \begin{align}
        % \underline{\alpha}(\lVert \xi_{x\upsilon}(t)-\xi_{x'\upsilon'}(t)\rVert)\leq V(\xi_{x\upsilon}(t),\xi_{x'\upsilon'}(t)) \nonumber\\
        \lVert x_{a\upsilon}(t)-x_{a'\upsilon'}(t)\rVert \leq \underline{\alpha}^{-1}(V(x_{a\upsilon}(t),x_{a'\upsilon'}(t))). \label{d_V_relation}
    \end{align}
    For any $t\in\R^+_0$, $\upsilon,\upsilon':\mathbb{R}_0^+ \rightarrow W$ and for all $x_{a\upsilon}(0):=a \in X,x_{a'\upsilon'}(0):=a'\in X$,
    \begin{align*}
        \dot V(x_{a\upsilon}(t),x_{a'\upsilon'}(t))\hspace{-0.2em}\leq\hspace{-0.5em}-\kappa V(x_{a\upsilon}(t),x_{a'\upsilon'}(t))\hspace{-0.3em}+\hspace{-0.3em}\sigma(\lVert \upsilon\hspace{-0.3em}-\hspace{-0.3em}\upsilon'\rVert_{\infty})\hspace{-0.3em}+\hspace{-0.3em}\Tilde{c}.
    \end{align*}
    By applying comparison lemma and by \eqref{d_V_relation} we get,
    \begin{align*}
        &V(x_{a\upsilon}(t),x_{a'\upsilon'}(t)) \leq e^{-\kappa t} V(a,a')\hspace{-0.3em}+\hspace{-0.3em}\frac{1}{\kappa}(\sigma(\lVert \upsilon\hspace{-0.3em}-\hspace{-0.3em}\upsilon'\rVert_{\infty})\hspace{-0.3em} +\hspace{-0.3em}\Tilde{c}),\\
        &\lVert x_{a\upsilon}(t)-x_{a'\upsilon'}(t)\rVert\hspace{-0.3em} \leq \hspace{-0.2em} \underline{\alpha}^{-1}\bigl(e^{-\kappa t} V(a,a')+\frac{1}{\kappa}(\sigma(\lVert \upsilon\hspace{-0.3em}-\hspace{-0.3em}\upsilon'\rVert_{\infty}) \hspace{-0.3em}+\hspace{-0.3em}\Tilde{c})\bigr).
    \end{align*}
By using a result similar to \cite[Exercise-4.35]{Khalil:1173048}, we can write
\begin{align*}
    &\lVert x_{a\upsilon}(t)-x_{a'\upsilon'}(t) \rVert\hspace{-0.3em} \leq \hspace{-0.3em}\underline{\alpha}^{-1}\bigl(3e^{-\kappa t} V(a,a')\bigr)+\underline{\alpha}^{-1}\bigl(\frac{3\Tilde{c}}{\kappa}\bigr)\\ & +\underline{\alpha}^{-1}\bigl(\frac{3}{\kappa}(\sigma(\lVert \upsilon-\upsilon'\rVert_{\infty})\bigr) \leq\beta(\lVert a\hspace{-0.25em}-\hspace{-0.25em}a'\rVert,t)\hspace{-0.25em}+\hspace{-0.25em}\gamma(\lVert \upsilon\hspace{-0.25em}-\hspace{-0.25em}\upsilon'\rVert_{\infty})\hspace{-0.25em}+\hspace{-0.25em}c,
\end{align*}
where $\beta(\lVert a-a'\rVert,t):=\underline{\alpha}^{-1}\bigl(3e^{-\kappa t}\overline{\alpha}(\lVert a-a'\rVert)\bigr)$, $\gamma(\lVert \upsilon-\upsilon'\rVert_{\infty}):=\underline{\alpha}^{-1}\bigl(\frac{3}{\kappa}(\sigma(\lVert \upsilon-\upsilon'\rVert_{\infty} )\bigr)$, $c:=\underline{\alpha}^{-1}\bigl(\frac{3\Tilde{c}}{\kappa}\bigr)$.
This implies that system $\Sigma$ is $\delta$-ISpS as in Definition \ref{def:deltaISpS}.
\end{proof}
% \begin{remark}
%     In definition~\ref{def:deltaISpSLyapunov}, the Lyapunov function is defined on the state-space $X$. But in general, the function will be defined for any $(x,y) \in (\R^d \times \R^d)$.
% \end{remark}
\section{System Description and Preliminaries}
In this work, we aim to design a feedback control scheme {(i.e.)} $u:=p(x,\upsilon)$ for the system \eqref{eqn:sys_del_isps}, that enforces $\delta$-ISpS properties for a class of partially unknown nonlinear control systems enforcing invariance on a compact set $X$.

\subsection{Nonlinear System in Strict Feedback Form}
We consider a nonlinear system in strict feedback form, which is a class of control systems with $f$ given as:
\begin{align}
    \dot{x}_1&=x_2,\nonumber\\
    \dot{x}_2&=f(x)+g(x)u,\label{eqn:ELsys}
\end{align}
where $x(t)=\left[x_1(t)^{\top},x_2(t)^{\top}\right]^{\top}\in X \subset \R^d$ is the state of the system, $X=X_1 \times X_2$, $x_1(t) \in X_1\subset\R^n$, $x_2(t) \in X_2\subset\R^n$, $u:\mathbb{R}^+_0\rightarrow U$ is the input signal. We use $f_i$ to represent the $i^{th}$ component of the vector function $f$, where $i\in I:=\left\{1,\ldots,{n}\right\}$. For brevity, whenever we use the subscript $i$, we refer to the whole set $I$ unless stated otherwise. 
% For clarity, in the subsequent discussions, particularly in Section~\ref{section:main theorem} where proofs involving two trajectories are presented, we adopt the notation $x, y: \mathbb{R}^+_0 \rightarrow X$ to denote the system trajectories originating from initial conditions $x(0), y(0) \in X$ for system \eqref{eqn:ELsys} (in place of $\xi$). 
Furthermore, to simplify the presentation when the context is unambiguous, the specific state $x(t)$ will be represented concisely as $x$. Now, we make the following assumptions on the system (\ref{eqn:ELsys}). 
\begin{assumption}\label{assume:dynamics}
    In system (\ref{eqn:ELsys}), we assume that  $f$ is unknown and the function $g$ (i.e., $g(x)$) is known. The matrix $g(x)^{-1}$ is assumed to exist for all $x \in X$.
\end{assumption}
This assumption is mild and broadly applicable, as it is naturally satisfied by many systems, including Euler–Lagrange systems. In order to deal with unknown parts of the dynamics, we utilize GP approximation (discussed in Subsection \ref{section:GP}), for which we need the following assumptions.
\begin{assumption}[\cite{HircheGP}]\label{assume:RKHS}
 The unknown function $f$ has a bounded reproducing kernel Hilbert space (RKHS) norm with respect to a known kernel $k$, that is $\lVert f_i\rVert_{k}<\infty$.
 \end{assumption}
 
 An RKHS for the kernels used is dense in the space of continuous functions restricted to a compact set $X$, which allows the kernels to approximate any continuous function on $X$. In order to train the model, we also need the following assumption on the availability of the dataset. 
\begin{assumption}[\cite{HircheGP}]\label{assume:measurements}
    The measurements $x\in X$ and $y=f(x)+w$ are accessible, where $w\sim \mathcal{N}(0_n,\rho_f^2\mathbf{I}_n)$ is an additive noise.
\end{assumption}
Practically, the measurements $f(x)$ can be approximated using state measurements obtained by running the system for a very small sampling time from different initial conditions with input signal $u \equiv0$. The noise $w$ is used to accommodate the approximations.
% \subsection{Problem Statement}
% \begin{problem}\label{statement}
%     Given the partially unknown control system $\Sigma$ given by (\ref{eqn:ELsys}) Assumptions \ref{assume:inertia}-\ref{assume:measurements}, the notion of incremental Input-to-State practical Stability presented in Definition \ref{def:deltaISpS} and its characterization defined in Definition \ref{def:deltaISpSLyapunov}, develop a back-stepping control design scheme that renders the control system $\Sigma$ $\delta$-ISpS.
% \end{problem}
% Problem \ref{statement} can be solved by first learning the unknown function $f$ using Gaussian Process and then developing a backstepping scheme based on the learned dynamics. The next section explains how Gaussian Process can be used to learn the unknown dynamics.
\vspace{-0.1cm}
\subsection{Gaussian Process Models}\label{section:GP}
Gaussian process can be used to approximately learn unknown nonlinear dynamics $f:X\rightarrow\R^{{n}}$ by using potentially noisy measurements, while also indicating model fidelity based on distance to training data. A GP is a stochastic process that assigns a joint Gaussian distribution to a finite subset $\{x^{(1)},\ldots,x^{(N)}\}\subset X$ \cite{GPBook}. Since $f$ is $n$-dimensional and GP only gives scalar values, each component $f_{i}$ is approximated with a GP:
\begin{align}
    \hat{f}_{i}(x)\sim\mathcal{GP}(\mu_{i}(x),k_{i}(x,x')),
\end{align}where $\mu_{i}:X \rightarrow \R$ is a mean function and $k_{i}:X\times X \rightarrow \R_0^+$ is a kernel that measures similarity between any $x,x'\in X$. For the prior mean function, any real-valued function can be used, but it is common practice to set $\mu_{i}(x)=0$ for all $x\in X$. On the contrary, the kernel function is problem-dependent. Some of the most commonly used kernels are linear, squared-exponential, and Mat\`ern kernels \cite{GPBook}. The approximation of $f$ is given by,
\begin{align}
    \hat{f}(x)=\begin{cases}
        \hat{f}_{1}(x)\sim\mathcal{GP}(0,k_{1}(x,x')),\\
        \hspace{1.5cm}\vdots\\
        \hat{f}_{{n}}(x)\sim\mathcal{GP}(0,k_{n}(x,x')).
    \end{cases}
\end{align}
Given a dataset $\mathcal{D}=\{(x^{(j)},y^{(j)})\}_{j=1}^N$, where $y^{(j)}=f(x^{(j)})+\omega^{(j)},\forall j\in\{1,\ldots,N\}$ following Assumption \ref{assume:measurements} and $N\in\N$, for an arbitrary state $x\in X$, the posterior distribution corresponding to $f_{i}(x)$ is computed as a normal distribution with mean $\mu_{i}(x)$ and covariance $\rho^2_{i}(x)$ as:
\begin{align}
    &\mu_{i}(x)=\Bar{k}_{i}^{\top}(K_{i}+\rho_f^2\mathbf{I}_N)^{-1}y_{i},\label{eqn:mean_i}\\
    &\rho^2_{i}(x)=k_{i}(x,x)-\Bar{k}_{i}^{\top}(K_{i}+\rho^2_f\mathbf{I}_N)^{-1}\Bar{k}_{i},\label{eqn:variance_i}
\end{align}where $\Bar{k}_{i}=[k_{i}(x^{(1)},x),\ldots,k_{i}(x^{(N)},x)]^{\top}\in\R^N$, $y_{i}=[y_{i}^{(1)},\ldots,y_{i}^{(N)}]^{\top}\in\R^N$ and $$K_{i}=\begin{bmatrix}k_{i}(x^{(1)},x^{(1)}) & \ldots & k_{i}(x^{(1)},x^{(N)})\\\vdots & \ddots & \vdots\\   k_{i}(x^{(N)},x^{(1)}) & \ldots & k_{i}(x^{(N)},x^{(N)})
\end{bmatrix}\in\R^{N\times N}.$$ Due to the continuity of the kernels, there exists a bound $\Bar{\rho}^2_{i}=\max_{x\in X}\rho^2_{i}(x)$. The overall function $f(x)$ is approximated by the distribution $\mathcal{N}(\mu(x),\rho(x))$, where
\begin{align}
    \label{eqn:mean}\mu(x):=[\mu_{1}(x),\ldots,\mu_{n}(x)]^{\top},\\
    \label{eqn:variance}\rho^2(x):=[\rho^2_{1}(x),\ldots,\rho^2_{n}(x)]^{\top}.
\end{align}
Using Assumption \ref{assume:RKHS}, the difference between the unknown $f(x)$ and the inferred mean $\mu(x)$ can be upper-bounded with high probability.
\begin{proposition}\label{prop:pobabilisticbound}
    Consider system (\ref{eqn:ELsys}), with assumptions \ref{assume:dynamics}-\ref{assume:measurements} and the learned GP model with mean $\mu$ and standard deviation $\rho$ as defined in (\ref{eqn:mean}) and (\ref{eqn:variance}), respectively, and $\Bar{\rho}^2_{i}=\max_{x\in X}\rho^2_{i}(x)$. The model error is bounded by
    \begin{align}
    \label{eqn:prob_for_gp}
        \mathbb{P}\left\{\lVert f(x)-\mu(x)\rVert\leq\lVert\eta\rVert\lVert\Bar{\rho}\rVert,\forall x\in X\right\}\geq (1-\epsilon)^n,
    \end{align}
    where $\epsilon\in(0,1)$, $\eta=\left[\eta_1,\ldots\eta_n\right]^{\top}$, $\eta_i=\sqrt{2\lVert f_i\rVert^2_{k_i}+300\gamma_i\log^3\left(\frac{N+1}{\epsilon}\right)}$, $N$ is the number of samples in $\mathcal{D}$, and $\gamma_i$ is the maximum information gain.
\end{proposition}
    \begin{proof}
        The proof follows from \cite[Lemma 2]{HircheGP}, where $\mathbb{P}\left\{\lVert f(x)-\mu(x)\rVert\leq\lVert\eta\rVert\lVert\rho(x)\rVert,\forall x\in X\right\}\geq (1-\epsilon)^n$. Similar to \cite{JagtapGP}, we can rewrite it as $\mathbb{P}\left\{\lVert f(x)-\mu(x)\rVert\in\{d|d\in[0,\lVert\eta\rVert\lVert\Bar{\rho}\rVert]\},\forall x \in X)\right\}\geq (1-\epsilon)^n$, where $\Bar{\rho}=[\Bar{\rho}_1,\ldots,\Bar{\rho}_n]^{\top}$ and $\Bar{\rho}_i:=\max_{x\in X}\rho_i(x)$. Hence, it is obvious that $\mathbb{P}(\lVert f(x)-\mu(x)\rVert\leq\lVert\eta\rVert\lVert\Bar{\rho}\rVert,\forall x\in X)\geq (1-\epsilon)^n$. For further discussion, refer to \cite{InfoGain}.
    \end{proof}
    %  \begin{remark}\label{remark:infogain}
    %     The information gain $\gamma_i$ used in Proposition \ref{prop:pobabilisticbound} quantifies the maximum mutual information between a subset of $\mathcal{D}$ and $f_i$. Its evaluation is NP-hard, but it can be greedily approximated. Moreover, $\gamma_i$ has a sublinear dependency on $N$ for most kernels, thus, even though $\eta$ increases, the model error decreases with more data points. 
    % \end{remark}
Leveraging the learned dynamics and the associated model error bounds, we proceed to design a controller using a feedback linearization framework, which is detailed in Section \ref{section:main theorem}. And this controller must ensure state-space invariance as we work with a compact state set $X$. To work with invariance, we introduce control barrier functions.
\subsection{Control Barrier Functions}
In this section, we formally introduce control barrier functions (CBFs) and their applications in the context of state-space invariance. Given the nonlinear control system $\Sigma$ in control affine form:
\begin{equation}
\label{general system equation}
    \Dot{x}=\Bar{f}(x)+\Bar{g}(x)\Bar{u},
\end{equation}
where, $\Bar{f}(x)=\begin{bmatrix}
    x_2\\
    f(x)
\end{bmatrix}$, $\Bar{g}(x)=\begin{bmatrix}
    0_{n\times n} & 0_{n \times m}\\
    0_{n\times n} & g(x)
\end{bmatrix}$, $\Bar{u}=\begin{bmatrix}
    0_n\\
    u
\end{bmatrix}\in \R^{(n+m)}$, $x\in {X} \subset \mathbb{R}^d$ is the state of system, and $u \in \mathbb{R}^m$. Assume that the functions $\Bar{f}:\mathbb{R}^d \rightarrow \mathbb{R}^d$ and $\Bar{g}:\mathbb{R}^d \rightarrow \mathbb{R}^d \times \mathbb{R}^{(n+m)}$ are continuously differentiable. Notice that $\Bar{f}$, $\Bar{g}$, and $\Bar{u}$ are used to differentiate from the notations of $f$, $g$ and $u$ used in \eqref{eqn:ELsys}. Consider a set $\mathcal{C}$ defined as the super-level set of a continuously differentiable function $h:{X} \rightarrow \mathbb{R}$ yielding,
\begin{equation}
\label{level sets equation}
    \left.
    \begin{aligned}
        \mathcal{C}&=\{x \in {X} \subset \mathbb{R}^d  : h(x) \geq 0 \}\\
        \partial \mathcal{C}&=\{x\in {X} \subset \mathbb{R}^d  : h(x) = 0 \}\\
        \text{Int}(\mathcal{C})&=\{x \in {X} \subset \mathbb{R}^d  : h(x) > 0 \}
    \end{aligned} \quad
    \right\}.
\end{equation}
It is assumed that $\text{Int}(\mathcal{C})$ is non-empty and $\mathcal{C}$ has no isolated points, \emph{(i.e.,)} $\text{Int}(\mathcal{C}) \neq \phi$ and $\overline{\text{Int}(\mathcal{C})} = \mathcal{C}$. Given a Lipschitz continuous control law $\Bar{u}=k(x)$, the resulting closed-loop system dynamics are $\Dot{x}=\Bar{f}_{cl}(x)=\Bar{f}(x)+\Bar{g}(x)k(x)$. The solution of the system is safe with respect to the control law $\Bar{u}=k(x)$ if the set $\mathcal{C}$ is forward invariant \textit{(i.e.)}
$\forall x(0) \in \mathcal{C} \implies x(t) \in \mathcal{C} \quad \text{for all time } t \geq 0.$ We can mathematically verify if the controller $k(x)$ is ensuring safety or not, by using control barrier functions (CBFs), which is defined next. 
\begin{Definition}[\cite{jagtap2020control}]
    \emph{(Control barrier function (CBF)):} Given the
set $\mathcal{C}$ defined by \eqref{level sets equation}, with $\frac{\partial h(x)}{\partial x} \neq 0$, for all $x \in \partial \mathcal{C}$, the function $h:X \subset \mathbb{R}^d \rightarrow \mathbb{R}$ is called the control barrier function (CBF), if there exists an extended class $\mathcal{K}$ function $\alpha$ such that for all $x \in X$, there exists input $\bar{u} \in U$, such that: $\dot{h}(x,\bar{u})+\alpha(h(x)) \geq 0$, 
where $\dot{h}(x,\bar{u})=\mathcal{L}_{\Bar{f}}h(x)+\mathcal{L}_{\Bar{g}}h(x)\bar{u}$, 
$\mathcal{L}_{\Bar{f}}h(x)=\frac{\partial h}{\partial x}\bar{f}(x)$, $\mathcal{L}_{\Bar{g}}h(x)=\frac{\partial h}{\partial x}\Bar{g}(x)$ are the Lie derivatives.
\end{Definition}
In our case, since the system dynamics is unknown, we use a variant of the control barrier function as given in \cite{jagtap2020control} to prove that the controller designed makes the system state-space invariant. In the next section, we present the controller that will make the system incrementally input to state practically stable, along with state-space invariance.
\section{Incremental Stability with Safety}
\label{section:main theorem}
In this section, we present the main result of the paper on the feedback linearization control design scheme, providing controllers that render the system $\delta$-ISpS, along with the system being state-space invariant.
% \begin{assumption}
% \label{assume:cbf}
%     There exists a function $h(x):\mathbb{R}^d \rightarrow \mathbb{R}$ which is a continuously differentiable function whose zero super-level set captures the entire state-space. The system dynamics $f(x)$ is approximated using the mean $\mu(x)$ and standard deviation $\rho(x)$ obtained from the Gaussian process which is trained over a slightly bigger set $\mathcal{C}_d$ than the considered state-space set $X$ given as:
%     \begin{equation}
% \label{eqn:level sets equation assumption}
%     \left.
%     \begin{aligned}
%         \mathcal{C}_d&=\{x \in \mathbb{R}^d  : h(x) \geq -\theta(\lVert d\rVert_{\infty} \},\\
%         \partial \mathcal{C}&=\{x\in  \mathbb{R}^d  : h(x) = -\theta(\lVert d\rVert_{\infty} )\},\\
%         \text{Int}(\mathcal{C})&=\{x \in \mathbb{R}^d  : h(x) > -\theta(\lVert d\rVert_{\infty}) \},
%     \end{aligned} \quad
%     \right.
% \end{equation}
% for some $\theta\in \mathcal{K}_{[0,a)}$ and $\lVert d \rVert_{\infty}\leq \Bar{d}\in [0,a)$ along with satisfying $\frac{\partial h(x)}{\partial x} \neq 0$, for all $x \in \partial \mathcal{C}$. 
% \end{assumption}
\begin{assumption}
\label{assume:cbf}
We assume the existence of a continuously differentiable function $h : X \subset\mathbb{R}^d \rightarrow \mathbb{R}$, referred to as a control barrier function (CBF), such that its zero superlevel set characterizes the admissible state space. The true system dynamics $f(x)$ (as given in \eqref{eqn:ELsys}) is approximated by a GP, which yields a mean function $\mu(x)$ and standard deviation $\rho(x)$, trained over an extended domain $\mathcal{C}_d \supset X$ that slightly enlarges the original state space $X$ to account for model uncertainty caused by GP. Specifically, we define $\mathcal{C}_d:= \left\{ x \in \mathbb{R}^d \;\middle|\; h(x) \geq -\chi\left( \lVert d \rVert_{\infty} \right) \right\}$, $\partial \mathcal{C}_d:= \left\{ x \in \mathbb{R}^d \;\middle|\; h(x) = -\chi\left( \lVert d \rVert_{\infty} \right) \right\}$, $\mathrm{Int}(\mathcal{C}_d):= \left\{ x \in \mathbb{R}^d \;\middle|\; h(x) > -\chi\left( \lVert d \rVert_{\infty} \right) \right\}$, where $\chi \in \mathcal{K}_{[0,a)}$ is a class-$\mathcal{K}$ function and $\lVert d \rVert_{\infty}:=\lVert{\eta}\rVert \lVert \Bar{\rho}\rVert \lVert \nabla_{x_2}h(x) \rVert_{\infty}$, where $\lVert{\eta}\rVert \lVert \Bar{\rho}\rVert$ is given by \eqref{eqn:prob_for_gp} along with satisfying $\frac{\partial h(x)}{\partial x} \neq 0$, for all $x \in \partial \mathcal{C}$. 
\end{assumption}

\begin{theorem}\label{theorem:Control}
    Consider nonlinear system $\Sigma$ given by (\ref{eqn:ELsys}) satisfying Assumptions \ref{assume:dynamics}-\ref{assume:measurements} and Assumption \ref{assume:cbf}, the function $f(x)$ approximated by a trained $\text{GP}$ with mean $\mu(x)$, and standard deviation $\rho(x)$ given by \eqref{eqn:mean} and \eqref{eqn:variance}, with $\Bar{\rho}:=\max_{x\in X}\rho(x)$. The feedback control law given by
    \begin{align}\label{eqn:controlLaw}
&u = g(x)^{-1}\hspace{-0.2em} \bigg[\hspace{-0.3em} -\hspace{-0.3em}\mu(x)\hspace{-0.3em} -\lambda_1x_2 - \lambda_2x_2 -\lambda_1\lambda_2x_1+\hspace{-0.2em} \upsilon+\hspace{-0.2em} \frac{ \text{ReLU}\big(\hspace{-0.2em}-\phi_0(x) \hspace{-0.2em}-\hspace{-0.2em} \phi_1(x) \upsilon\big) \hspace{-0.1em}\cdot\hspace{-0.1em}\phi_1(x)^\top}{\phi_1(x) \phi_1(x)^\top}\hspace{-0.3em}  \bigg]
\end{align}
ensures that the system \eqref{eqn:ELsys} is $\delta$-ISpS with respect to the nominal input signal $\upsilon$ and simultaneously is invariant within the set $\mathcal{C}_d$ with probability at least $(1 - \epsilon)^n$, where $\lambda_1 > \tfrac{1}{2}, \quad \lambda_2 > \hspace{-0.2em}2+\frac{1}{2\theta},$ for some $\theta>0$, $\phi_0(x) = \nabla_{x_1} h(x)\, x_2 + \nabla_{x_2} h(x)\left(-\lambda_1x_1 - \lambda_2x_2 -\lambda_1\lambda_2x_1 \right) + \alpha(h(x))$, and $\phi_1(x) = \nabla_{x_2} h(x)$.
\end{theorem}
\begin{proof}
As we consider a compact state-space $X$, we begin by establishing the robust state-space invariance of the proposed controller through the framework of CBF \cite{jagtap2020control}. Traditional CBF methods require explicit knowledge of the system dynamics. However, since the true dynamics $f(x)$ are unknown, we instead rely on their approximation using a GP, which provides a mean function $\mu(x)$ and a standard deviation $\rho(x)$, valid over the compact domain $X$. For notational simplicity, we denote $\phi_0(x)$ and $\phi_1(x)$ as $\phi_0$ and $\phi_1$ respectively in the following steps. To prove the controller ensures invariance in the set $\mathcal{C}_d$ we consider,
\begin{align*}
    &\dot{h}(x)+\alpha(h(x))=\nabla_{x_1}h(x)\Dot{x}_1+\nabla_{x_2}h(x)\Dot{x}_2+\alpha(h(x))\\
    &=\nabla_{x_1} h(x)x_2\hspace{-0.2em}+\hspace{-0.2em}\alpha(h(x))\hspace{-0.2em}+\hspace{-0.2em}\nabla_{x_2}h(x)(\mu(x)\hspace{-0.2em}+\lVert\eta\rVert\lVert\Bar\rho\rVert+g(x)\biggl(g(x)^{-1}\bigl(-\mu(x)-\lambda_1x_2\hspace{-0.25em}-\hspace{-0.25em}\lambda_2x_2\hspace{-0.25em}-\hspace{-0.25em}\lambda_1\lambda_2 x_1+\upsilon\\&+\frac{\text{ReLU}(-\phi_0\hspace{-0.25em}-\hspace{-0.25em}\phi_1\upsilon)\hspace{-0.15em}\phi_1^\top}{\phi_1\phi_1^\top}\bigr)\biggr)
    % &= \nabla_{x_1} h(x)x_2\hspace{-0.2em}+\hspace{-0.2em}\alpha(h(x))\hspace{-0.2em}+\hspace{-0.2em}\nabla_{x_2}h(x)(-\lambda_1x_2\hspace{-0.25em}-\hspace{-0.25em}\lambda_2x_2\hspace{-0.25em}-\hspace{-0.25em}\lambda_1\lambda_2 x_1\\&+\upsilon+\frac{\text{ReLU}(-\phi_0\hspace{-0.25em}-\hspace{-0.25em}\phi_1\upsilon)\hspace{-0.15em}\phi_1^\top}{\phi_1\phi_1^\top})+\nabla_{x_2}h(x)\lVert{\eta}\rVert \lVert \Bar{\rho}\rVert\\
    =\nabla_{x_1}h(x) x_2\hspace{-0.2em}\hspace{-0.2em}+\hspace{-0.2em}\alpha(h(x))\hspace{-0.2em}+\hspace{-0.2em}\nabla_{x_2}h(x)(-\lambda_2x_2\hspace{-0.25em}-\hspace{-0.25em}\lambda_1x_1\hspace{-0.25em}-\hspace{-0.25em}\lambda_1\lambda_2 x_1)+\hspace{-0.2em}\nabla_{x_2}h(x)\upsilon\\&+\hspace{-0.2em}\nabla_{x_2}h(x)\biggl(\hspace{-0.3em}\frac{\text{ReLU}(-\phi_0\hspace{-0.3em}-\hspace{-0.3em}\phi_1\upsilon)\hspace{-0.15em}\phi_1^\top}{\phi_1\phi_1^\top}\hspace{-0.3em}\biggr)\hspace{-0.2em}+\hspace{-0.2em}\chi(\lVert d\rVert_{\infty})\\
&=\phi_0\hspace{-0.3em}+\hspace{-0.3em}\phi_1\upsilon+\hspace{-0.2em}\nabla\hspace{-0.25em}_{x_2}h(x)\biggl(\hspace{-0.3em}\frac{\text{ReLU}(-\phi_0\hspace{-0.3em}-\hspace{-0.3em}\phi_1\upsilon)\phi_1^\top}{\phi_1\phi_1^\top}\hspace{-0.3em}\biggr)\hspace{-0.2em}+\hspace{-0.2em}\chi(\lVert d\rVert_{\infty})=\phi_0\hspace{-0.3em}+\hspace{-0.3em}\phi_1\upsilon+\hspace{-0.2em}\text{ReLU}(-\phi_0\hspace{-0.3em}-\hspace{-0.3em}\phi_1\upsilon)\biggl(\hspace{-0.3em}\frac{\phi_1\phi_1^\top}{\phi_1\phi_1\hspace{-0.25em}^\top}\hspace{-0.3em}\biggr)\hspace{-0.2em}+\hspace{-0.2em}\chi(\lVert d\rVert_{\infty})\\
&\geq \phi_0\hspace{-0.3em}+\hspace{-0.3em}\phi_1\upsilon\hspace{-0.3em}+\hspace{-0.2em}\text{ReLU}(-\phi_0-\phi_1\upsilon)\hspace{-0.2em}-\hspace{-0.2em}\chi(\lVert d\rVert_{\infty}),
\end{align*}
where $\chi:=\mathcal{I}_d \in \mathcal{K}$, $\lVert d\rVert_{\infty}:=\lVert{\eta}\rVert \lVert \Bar{\rho}\rVert \lVert \nabla_{x_2}h(x) \rVert_{\infty}$.
Here, one can have two cases. If $-\phi_0-\phi_1\upsilon\geq0$, then $\dot{h}(x) +\alpha(h(x)) \geq -\chi(d)$ and if $-\phi_0-\phi_1\upsilon<0$, then $\dot{h}(x)+\alpha(h(x))\geq\phi_0(x)+\phi_1(x)\upsilon-\chi(d)>-\chi(d)$. Thus, in both cases, $\dot{h}(x) + \alpha(h(x)) > -\chi(d)$,
which guarantees that the CBF condition is satisfied up to a bounded relaxation. Consequently, by invoking the results from \cite[Theorem 2]{kolathaya2018input}, one can conclude that the proposed control law ensures robust forward invariance of the set $\mathcal{C}_d := \left\{ x \in \mathbb{R}^d \; \middle| \; h(x) + \chi(\lVert d \rVert_{\infty}) \geq 0 \right\}.$

Next, we show that the control input given by \eqref{eqn:controlLaw} renders the system $\delta$-ISpS with respect to the external input $\upsilon$. From Theorem \ref{theorem:ISpSLyapunov}, it is sufficient to show that a Lyapunov function as defined in Definition \ref{def:deltaISpSLyapunov} exists for the closed-loop control system. Consider the $\delta-$ISpS Lyapunov function, $V(x,y)\hspace{-0.25em}=\hspace{-0.25em}\frac{1}{2}\lVert x_1-y_1\rVert^2+\frac{1}{2}\lVert x_2+\lambda_1x_1-y_2-\lambda_1y_1\rVert^2$. We define $e_1=x_1-y_1$, and $e_2=x_2+\lambda_1x_1-y_2-\lambda_1y_1$ and the Lyapunov function is $V(e_1,e_2)=\frac{1}{2}\lVert e_1\rVert^2+\frac{1}{2}\lVert e_2\rVert^2$. For the chosen Lyapunov function to satisfy Definition \ref{def:deltaISpSLyapunov}, we start showing that the Lyapunov function is upper bounded by $\bar{\alpha}(\lVert x-y \rVert)$.
\begin{align*}
    &V(x,y)=\frac{1}{2}\lVert x_1-y_1\rVert^2+\frac{1}{2}\lVert x_2+\lambda_1x_1-y_2-\lambda_1y_1\rVert^2\\&\leq \frac{1}{2}\lVert x_1-y_1\rVert^2+\lVert x_2-y_2\rVert^2+\lambda_1^2\lVert x_1-y_1\rVert^2\\
    % &\leq \bigl(\frac{1+\lambda_1}{2}\bigr)(\lVert x_1-y_1\rVert^2+\lVert x_2-y_2\rVert^2)\\
    &\leq \max{\biggl(\frac{1}{2}+\lambda_1^2,1\biggr)}\lVert x-y \rVert^2:=\bar{\alpha}(\lVert x-y \rVert).
\end{align*}
Now we show that the chosen Lyapunov function is lower bounded by $\underline{\alpha}(\lVert x-y \rVert)$.
\begin{align*}
    &V(x,y)=\frac{1}{2}\lVert x_1-y_1\rVert^2+\frac{1}{2}\lVert x_2+\lambda_1x_1-y_2-\lambda_1y_1\rVert^2\\
    &= \frac{1}{2}(\lVert x_1-y_1\rVert^2+\lVert x_2-y_2\rVert^2+\lambda_1^2\lVert x_1-y_1\rVert^2 +2\lambda_1(x_2-y_2)^\top (x_1-y_1))  \\& \geq \frac{1}{2}\bigl((1+\lambda_1^2)\lVert x_1-y_1\rVert^2+\lVert x_2-y_2\rVert^2-2\lambda_1\lVert x_1-y_1\rVert \lVert x_2-y_2\rVert \\
    &\geq \frac{1}{2} \lVert x-y \rVert^2:=\underline{\alpha}(\lVert x-y \rVert).
\end{align*}
Condition~(\ref{item:lyapunov1}) of Definition~\ref{def:deltaISpSLyapunov} is satisfied. Next, to check if the Lyapunov function satisfies condition~(\ref{item:lyapunov2}), we consider $\dot{e}_1=\dot{x}_1-\dot{y}_1=x_2-y_2$ and $e_2=x_2-y_2+\lambda_1 e_1$. This leads to $\dot{e}_1=x_2-y_2=e_2-\lambda_1 e_1$ and, $e_2=x_2+\lambda_1x_1-y_2-\lambda_1y_1$, $\dot{e}_2=(f(x)+g(x)u_1-f(y)-g(y)u_2)+\lambda_1(x_2-y_2)$. Thus, $\dot{V}(e_1,e_2)=e_1^\top \dot{e}_1+e_2^\top \dot{e}_2$. With proper substitution of error derivatives, the controller as in \eqref{eqn:controlLaw}, and applying Young's inequality,
\begin{align}
&\dot{V}(e_1,e_2)=e_1^\top (e_2-\lambda_1 e_1)+e_2^\top (f(x)+g(x)u_1-f(y)-g(y)u_2)+\lambda_1 e_2^\top(x_2-y_2),\nonumber\\
    &=e_1^\top e_2 -\lambda_1 \lVert e_1 \rVert^2+ e_2^\top(f(x)+g(x)g(x)^{-1}(-\mu(x)-(\lambda_1+\lambda_2)x_2-\lambda_1\lambda_2 x_1+\upsilon+\text{ReLU}(\phi_0-\phi_1\upsilon )\frac{\phi_1^\top}{\phi_1 \phi_1^\top})\nonumber\\&-f(y)-g(y)g(y)^{-1}(-\mu(y)-(\lambda_1+\lambda_2)y_2-\lambda_1\lambda_2 y_1+\upsilon'+\text{ReLU}(\phi_0'-\phi_1'\upsilon')\frac{\phi_1'^\top}{\phi_1'\phi_1'^\top}))+\lambda_1 e_2^\top(x_2-y_2),\nonumber \\
    % &=e_1^\top e_2 -\lambda_1 \lVert e_1 \rVert^2+e_2^\top(f(x)-\mu(x)-(\lambda_1+\lambda_2)x_2\nonumber\\&-\lambda_1\lambda_2 x_1+\upsilon+\text{ReLU}(\phi_0-\phi_1\upsilon )\frac{\phi_1^\top}{\phi_1 \phi_1^\top}\nonumber\\&-f(y)+\mu(y)+(\lambda_1+\lambda_2)y_2+\lambda_1\lambda_2 y_1-\upsilon'\nonumber\\&-\text{ReLU}(\phi_0'-\phi_1'\upsilon')\frac{\phi_1'^\top}{\phi_1'\phi_1'^\top})+\lambda_1 e_2^\top(x_2-y_2),\nonumber \\
    &=e_1^\top e_2 -\lambda_1 \lVert e_1 \rVert^2+e_2^\top (f(x)-\mu(x))+e_2^\top (\mu(y)-f(y))-(\lambda_1+\lambda_2)e_2^\top (x_2\hspace{-0.25em}-\hspace{-0.25em}y_2)\hspace{-0.25em}-\hspace{-0.25em}\lambda_1\lambda_2  e_2^\top (x_1\hspace{-0.25em}-\hspace{-0.25em}y_1)\hspace{-0.25em}+\hspace{-0.25em}e_2^\top(\upsilon\hspace{-0.25em}-\hspace{-0.25em}\upsilon')\nonumber\\&+e_2^\top (\text{ReLU}(\phi_0\hspace{-0.25em}-\hspace{-0.25em}\phi_1\upsilon )\frac{\phi_1^\top}{\phi_1 \phi_1^\top}-\text{ReLU}(\phi_0'\hspace{-0.25em}-\hspace{-0.25em}\phi_1'\upsilon')\frac{\phi_1'^\top}{\phi_1'\phi_1'^\top})+\lambda_1 e_2^\top(e_2\hspace{-0.25em}-\hspace{-0.25em}\lambda_1 e_1)\nonumber
    % ,\nonumber\\
    % &=e_1^\top e_2 -\lambda_1 \lVert e_1 \rVert^2+e_2^\top (f(x)-\mu(x))+e_2^\top (\mu(y)-f(y))\nonumber\\&+e_2^\top(-(\lambda_1+\lambda_2) (e_2-\lambda_1 e_1)-\lambda_1\lambda_2 e_2^\top e_1+\lambda_1 (e_2-\lambda_1 e_1))\nonumber\\&+e_2^\top (\text{ReLU}(\phi_0-\phi_1\upsilon )\frac{\phi_1^\top}{\phi_1 \phi_1^\top}-\text{ReLU}(\phi_0'-\phi_1'\upsilon')\frac{\phi_1'^\top}{\phi_1'\phi_1'^\top})\nonumber\\&+e_2^\top(\upsilon-\upsilon'),\nonumber\\
     \\&\leq e_1^\top e_2 \hspace{-0.3em}-\hspace{-0.3em}\lambda_1 \lVert e_1 \rVert^2 \hspace{-0.3em}+\hspace{-0.3em} \lVert e_2 \rVert \lVert f(x)\hspace{-0.3em}-\hspace{-0.3em}\mu(x) \rVert   + \lVert e_2 \rVert \lVert \mu(y)-f(y) \rVert -\lambda_2 \lVert e_2 \rVert^2 + e_2^\top(\upsilon-\upsilon')\nonumber\\& +e_2^\top (\text{ReLU}(\phi_0-\phi_1\upsilon )\frac{\phi_1^\top}{\phi_1 \phi_1^\top}-\text{ReLU}(\phi_0'-\phi_1'\upsilon')\frac{\phi_1'^\top}{\phi_1'\phi_1'^\top}),\nonumber\\
    & \leq  -\lambda_1 \lVert e_1 \rVert^2 -\lambda_2 \lVert e_2 \rVert^2 + \frac{1}{2}(\lVert e_1 \rVert^2+\lVert e_2 \rVert^2)+\frac{1}{2}(\lVert e_2 \rVert^2+\lVert f(x)\hspace{-0.25em}-\hspace{-0.25em}\mu(x) \rVert^2)+\frac{1}{2}(\lVert e_2 \rVert^2\hspace{-0.25em}+\hspace{-0.25em}\lVert \mu(y)\hspace{-0.25em}-\hspace{-0.25em}f(y) \rVert^2)\nonumber\\&+\frac{1}{2}(\lVert e_2 \rVert^2 + \lVert \upsilon -\upsilon'\rVert^2)+\frac{1}{2\theta}\lVert e_2 \rVert^2+ \frac{\theta}{2} \lVert \text{ReLU}(\phi_0-\phi_1\upsilon )\frac{\phi_1^\top}{\phi_1 \phi_1^\top}\hspace{-0.25em}-\hspace{-0.25em}\text{ReLU}(\phi_0'-\phi_1'\upsilon')\frac{\phi_1'^\top}{\phi_1'\phi_1'^\top} \rVert^2,\nonumber\\
    &\leq -(\lambda_1\hspace{-0.25em}-\hspace{-0.25em}\frac{1}{2})\lVert e_1 \rVert^2\hspace{-0.25em} -\hspace{-0.25em}(\lambda_2-2\hspace{-0.25em}-\hspace{-0.25em}\frac{1}{2\theta})\lVert e_2 \rVert^2+\frac{1}{2}\lVert \upsilon\hspace{-0.25em} -\hspace{-0.25em}\upsilon'\rVert^2 \hspace{-0.25em}+\hspace{-0.25em}\frac{1}{2}(\lVert\eta\rVert^2\lVert\Bar{\rho}\rVert^2\hspace{-0.35em}+\hspace{-0.35em}\lVert\eta\rVert^2\lVert\Bar{\rho}\rVert^2)\hspace{-0.25em}+\hspace{-0.25em}2\theta \hspace{-1em}\sup_{x \in X, \upsilon \in W}  \hspace{-1em}\lVert \text{ReLU}(\phi_0\hspace{-0.25em}-\hspace{-0.25em}\phi_1\upsilon )\frac{\phi_1^\top}{\phi_1 \phi_1^\top}\rVert^2,\nonumber\\
    &\leq -\kappa_1\lVert e_1 \rVert^2-\kappa_2 \lVert e_2 \rVert^2+\frac{1}{2}\lVert \upsilon -\upsilon'\rVert^2 +\lVert\eta\rVert^2\lVert\Bar{\rho}\rVert^2\hspace{-0.25em}+\hspace{-0.25em}2\theta \sup_{x \in X, \upsilon \in W} \lVert \text{ReLU}(\phi_0\hspace{-0.25em}-\hspace{-0.25em}\phi_1\upsilon )\frac{\phi_1^\top}{\phi_1 \phi_1^\top} \rVert^2,\nonumber\\
    &\leq -\kappa (\frac{1}{2}\lVert e_1 \rVert^2+\lVert e_2 \rVert^2)+ \sigma(\lVert \upsilon -\upsilon'\rVert) +\Tilde{c},\nonumber\\
    &\leq -\kappa V(e_1,e_2)+ \sigma(\lVert \upsilon -\upsilon'\rVert) +\Tilde{c}, \label{proof:lyapunovIneq}
\end{align}
where $\theta>0$, $\lambda_1 \geq \frac{1}{2}$, $\lambda_2\geq 2+\frac{1}{2\theta}$, $\kappa_1=\lambda_1-\frac{1}{2}$, $\kappa_2=\lambda_2-2-\frac{1}{2\theta}$, $\kappa:=2\min\{\kappa_1,\kappa_2\}$, $\sigma(r):=\frac{1}{2}r^2$, and $\Tilde{c}:=\lVert\eta\rVert^2\lVert\Bar{\rho}\rVert^2+2\theta \sup_{x \in X, \upsilon \in W} \lVert \text{ReLU}(\phi_0-\phi_1\upsilon )\frac{\phi_1^\top}{\phi_1 \phi_1^\top}\rVert$.
Thus \eqref{proof:lyapunovIneq} satisfies condition $\eqref{item:lyapunov2}$ of Definition \ref{def:deltaISpSLyapunov}. Thus, the Lyapunov function $V$ satisfies all conditions of Definition~\ref{def:deltaISpSLyapunov}. Thus, the system is $\delta$-ISpS.  Now, from Theorem~\ref{theorem:ISpSLyapunov}, the control law renders the system $\delta$-ISpS with respect to the input $\upsilon$. Now using Theorem~\ref{theorem:ISpSLyapunov}, we get:
   \begin{align*}
        \lVert x(t)-y(t))\rVert \leq\beta(\lVert x(0)-y(0)\rVert,t)+\gamma(\lVert {\upsilon}-{\upsilon}'\rVert_\infty)+c,
    \end{align*}
where $\beta(r,s)=\underline{\alpha}^{-1}(3e^{-\kappa s}\overline{\alpha}(r))
    =\sqrt{6\max\{0.5+\lambda_1^2,1\}r^2e^{-\kappa s}}$, $\gamma(r)=\underline{\alpha}^{-1}\left(\frac{3}{\kappa}{\sigma}(r)\right)=\sqrt{\frac{6}{\kappa}{\sigma}(r)}, \forall r,s\in\R^+_0$, $c=\underline{\alpha}^{-1}\left(\frac{3\Tilde{c}}{\kappa}\right)=\sqrt{\frac{6\Tilde{c}}{\kappa}}$ are the class-$\mathcal{KL}$, class-$\mathcal{K}_\infty$ functions and constant, respectively. Since the model error between the approximation and the actual function is bounded probabilistically, the system is made $\delta$-ISpS with the same high probability $(1-\epsilon)^n$.

This concludes that the control law in ~\eqref{eqn:controlLaw} renders the control system $\Sigma$, given by~\eqref{eqn:ELsys}, $\delta$-ISpS with respect to the input $\upsilon$ while ensuring invariance in the set $\mathcal{C}_d$.
\end{proof}

%=\underline{\alpha}^{-1}\left(\frac{3}{k}\lVert\eta\rVert^2\lVert\Bar{\rho}\rVert^2\right)

\begin{remark}
    Since we have a nonzero value for $\Tilde{c}$, a non-vanishing perturbation that produces a mismatch in trajectory even after an arbitrarily long time can be quantized to be ${c}$, where $\underline{\alpha}\in\mathcal{K}_\infty$. This signifies that even at an arbitrarily large value of $t$, the trajectories might not exactly converge to each other but might differ by a value $c$. The control law \eqref{eqn:controlLaw}, structured as $u = u_{\text{nom}} + \Delta u$, is an analytical solution to the CBF optimization problem. The $\Delta u$ term (the $\text{ReLU}$ function) serves as a minimal safety filter, ensuring the invariance constraint $\dot{h} + \alpha(h) \geq -\chi(d)$ is satisfied. This design guarantees that invariance is maintained while the $\delta$-ISpS stability objective is preserved.
\end{remark}

\section{Case Study-Two-link manipulator}\label{section:study}
The dynamics of a two link manipulator  \cite{2Regbook}  can be written as, $\dot \xi_1=\xi_2$, $\dot \xi_2 = {M^{-1}(\xi_1)\left[-H(\xi_1,\xi_2)-c(\xi_1)\right]}+{M^{-1}(\xi_1)}\tau,$ where 
% Here, $M(\xi)=ml^2\begin{bmatrix}
%     \left(\frac{5}{3}+\cos{x_2}\right) & \left(\frac{1}{3}+\frac{1}{2}\cos{x_2}\right)\\
%     \left(\frac{1}{3}+\frac{1}{2}\cos{x_2}\right) & \frac{1}{3}
% \end{bmatrix},$ $H(x_1,x_2) = ml^2\sin{x_2}\begin{bmatrix}
%     -\frac{1}{2}\dot{x}_2^2-\dot{x}_1\dot{x}_2\\
%     \frac{1}{2}\dot{x}_1^2
% \end{bmatrix}$, $c(x_1)=ma_gl\begin{bmatrix}
%     \frac{3}{2}\cos{x_1}+\frac{1}{2}\cos{(x_1+x_2)}\\
%     \frac{1}{2}\cos{(x_1+x_2)}
% \end{bmatrix}$
$\xi=[\xi_1^{\top},\xi_2^{\top}]^{\top}$, $\xi_1=\mathbf{p}$, $\xi_2=\dot{\mathbf{p}}$, $\mathbf{p}(t)=[x_1,x_2]^{\top}$, $x_1$ and $x_2$ are the angles of the two revolute joints and $\tau$ represents the torque inputs to the joints. The parameter matrices are chosen as in \cite{2Regbook}. The mass and length of both the links are $m=1kg$, and $l=1m$. It is obvious that Assumptions \ref{assume:RKHS}-\ref{assume:measurements} hold for $f$. We consider a compact set as state space given by $X=\{x \in \left[-\frac{\pi}{2},\frac{\pi}{2}\right]\times\left[-\frac{\pi}{2},\frac{\pi}{2}\right]\times\left[-0.2,0.2\right]\times\left[-0.2,0.2\right] | h(x):=1-\sqrt[p]{\bigl(\frac{2x_1}{\pi}\bigr)^p+\bigl(\frac{2x_2}{\pi}\bigr)^p+\bigl(\frac{x_3}{0.2}\bigr)^p+\bigl(\frac{x_4}{0.2}\bigr)^p} \geq 0\}$,  where $p$ is chosen as 20, to almost approximate the state space as a hypercube instead of a hyper-ellipsoid. Thus Assumption~\ref{assume:cbf} is satisfied. We train the unknown model $f$ using a Gaussian process with $800$ data samples of $x$ and $y=f(x)+w$, where $w\sim\mathcal{N}(0,\rho_f^2\mathbf{I}_2)$, $\rho_f=0.01$, collected by simulating the system with several randomly selected initial states. The considered kernel is $k_i(x,x')=\rho_{k_i}^2\exp{\left(\sum_{j=1}^4\frac{(x_j-x'_j)^2}{-2l^2_{ij}}\right)}$, $i=1,2$, where $\rho_{k_1}=115$, $\rho_{k_2}=186$, $l_{11}=1.54$, $l_{12}=0.541$, $l_{13}=136$, $l_{14}=120$, $l_{21}=1.77$, $l_{22}=0.489$, $l_{23}=122$ and $l_{24}=131$. We computed the mean and variance as shown in (\ref{eqn:mean}) and (\ref{eqn:variance}) with $\lVert\Bar{\rho}\rVert=0.366$. For a value of $\lVert\eta\rVert\lVert\Bar{\rho}\rVert=0.19$, the probability interval is $[0.9803,0.9822]$ with a confidence of $1-10^{-10}$, $25^4$ realizations. We designed the controller as shown in Theorem \ref{theorem:Control} with the values of $\theta=0.001$, $\lambda_1 =1.5$, $\lambda_2=503$, we get ${c}=8.6$. Fig.~\ref{fig:2Rtraj} shows the evolution of the system starting at two different initial conditions, which are converging towards each other. Fig.~\ref{fig:2Rdiff} shows the closeness of the trajectories and the bounds on the closeness for the system starting at two different initial conditions and two different inputs. 
\begin{figure}
\vspace{0.7em}
\centering
    \includegraphics[width=0.75\linewidth]{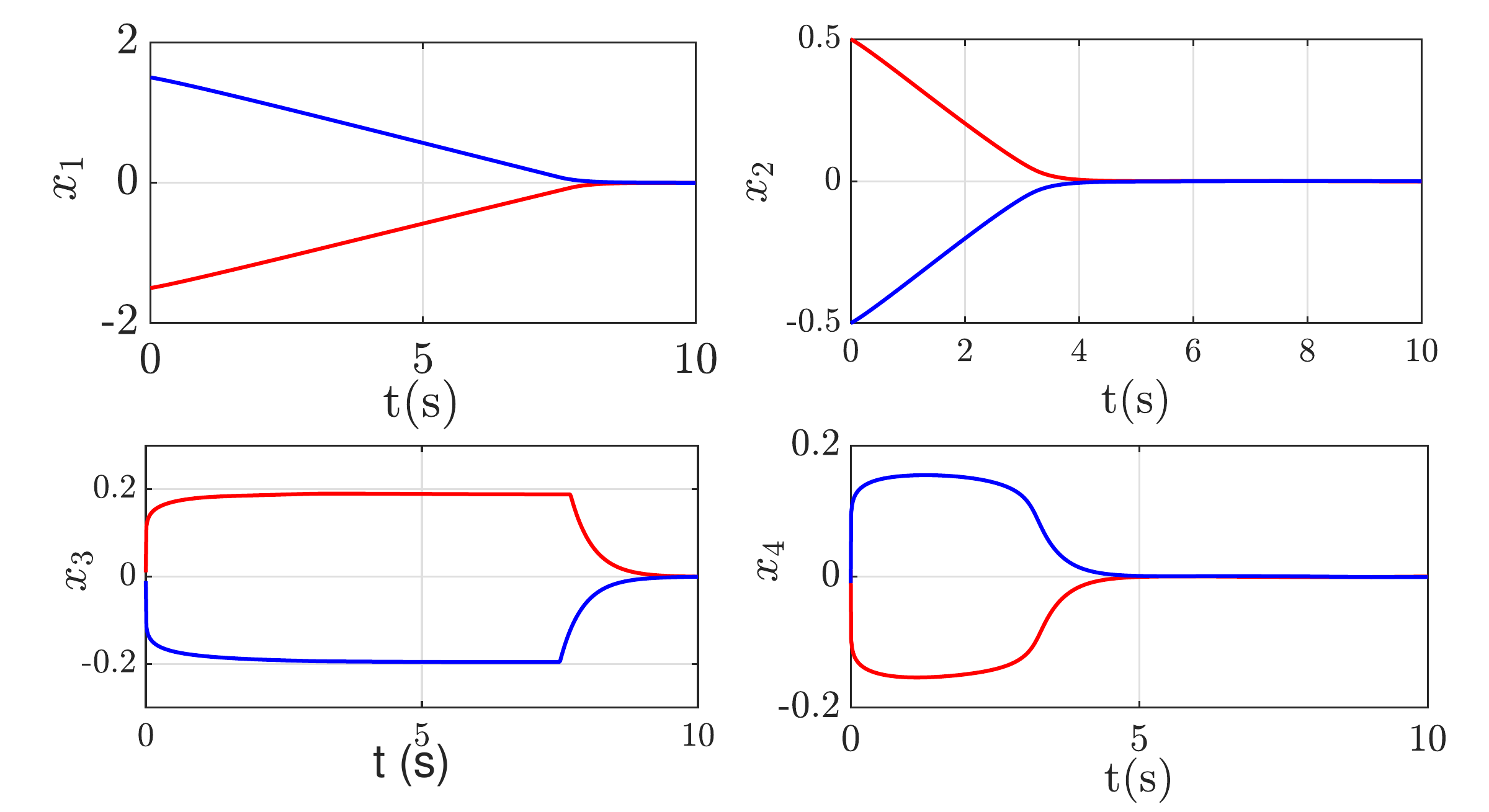}
    \caption{Evolution of the system under a inputs $\upsilon=[\sin(t),\cos(t)]$ and $\upsilon'=[\cos(t),\sin(t)]$ with the initial states $x_0=[-1.5,0.5,0.01,0.01]$ (blue line) and $x_0=[1.5,-0.5,-0.01,-0.01]$ (orange line).}
    \label{fig:2Rtraj}
\end{figure}
\begin{figure}
\centering
    \includegraphics[width=0.55\linewidth]{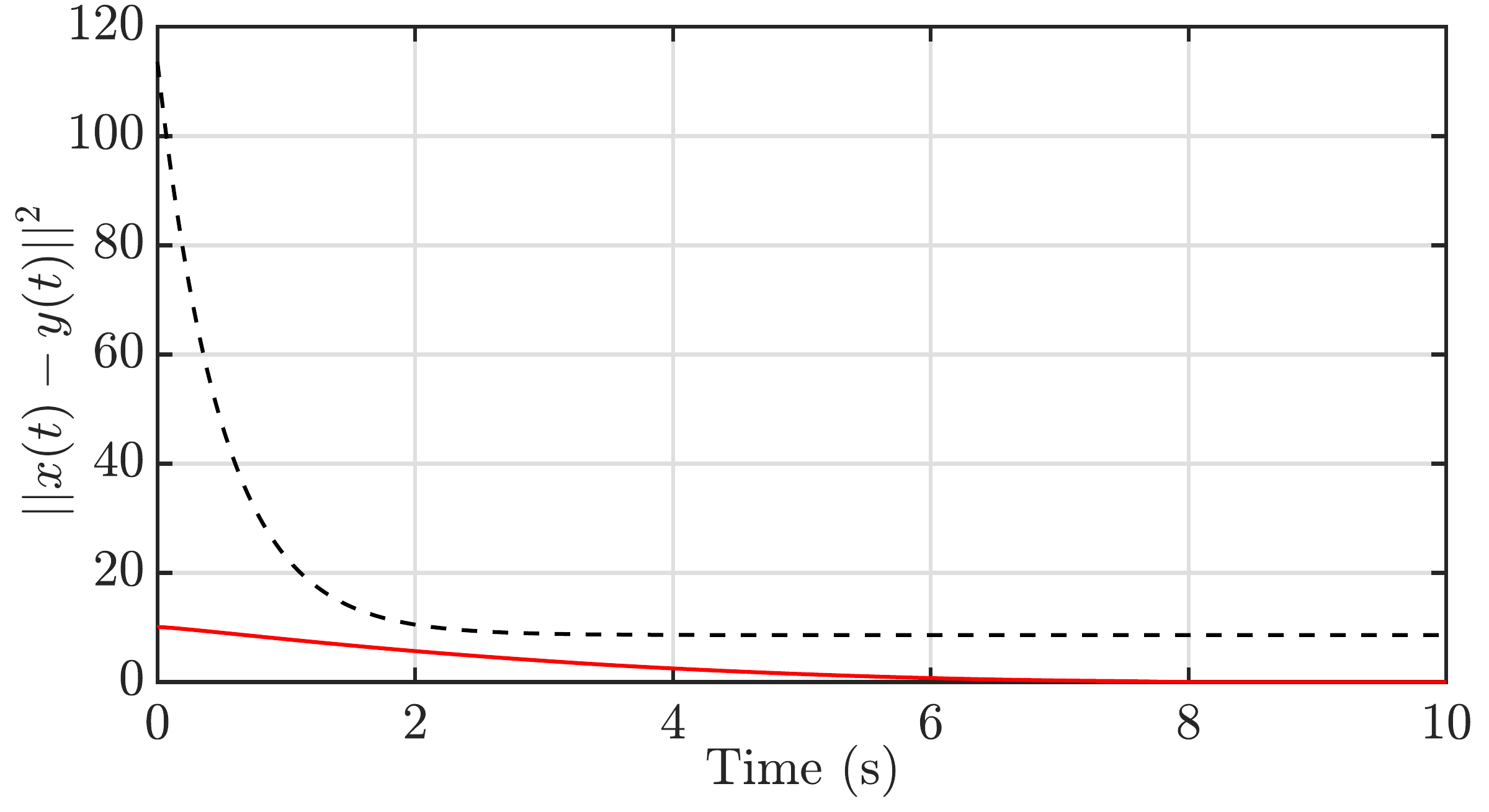}
    \caption{Distance between the trajectories of the controlled system under the inputs $\upsilon=[\sin(t),\cos(t)]$ and $\upsilon'=[\cos(t),\sin(t)]$ with the initial conditions $x_0=[-1.5,0.5,0.01,0.01]$ and $y_0=[1.5,-0.5,-0.01,-0.01]$ respectively.
    }
    \label{fig:2Rdiff}
\end{figure}
%%%%%%%%%%%%%%%%%%%%%%%%%%%%%%%%%%%%%%%%%%%%%%%%%%%%%%%%%%%%%%%%%%%%
\section{Conclusion}\label{section:conclude}
In this paper, we have presented a control design scheme for partially unknown systems that synthesizes controllers ensuring incremental input-to-state practical stability ($\delta$-ISpS) simultaneously guaranteeing invariance. We have defined the notion of $\delta$-ISpS and have also provided a Lyapunov function-based characterization for the same. The case study we have presented demonstrates both the convergence of system trajectories and establishes bounds on the closeness of the system trajectories starting from different initial states.
%%%%%%%%%%%%%%%%%%%%%%%%%%%%%%%%%%%%%%%%%%%%%%%%%%%%%%%%%%%%%%%%%%%%
%\section*{APPENDIX}

%\cite{knuth:1984}

%\section*{ACKNOWLEDGMENT}
%This work was supported in part by the Google Research Grant and the CSR Grant by Nokia Corporation.
%%%%%%%%%%%%%%%%%%%%%%%%%%%%%%%%%%%%%%%%%%%%%%%%%%%%%%%%%%%%%%%%%%%%%

\bibliographystyle{ieeetr} % We choose the "plain" reference style
\bibliography{sources} % Entries are in the refs.bib file

\begin{thebibliography}{10}

\bibitem{synchComplex}
G.~Russo and M.~di~Bernardo, ``Contraction theory and master stability function: Linking two approaches to study synchronization of complex networks,'' {\em IEEE Transactions on Circuits and Systems II: Express Briefs}, vol.~56, no.~2, pp.~177--181, 2009.

\bibitem{synchOsci}
G.-B. Stan and R.~Sepulchre, ``Analysis of interconnected oscillators by dissipativity theory,'' {\em IEEE Transactions on Automatic Control}, vol.~52, no.~2, pp.~256--270, 2007.

\bibitem{InterIncre}
B.~S. Dey, I.~N. Kar, and P.~Jagtap, ``On incremental stability of interconnected switched systems,'' {\em International Journal of Systems Science}, vol.~0, no.~0, pp.~1--16, 2025.

\bibitem{girard2014approximately}
A.~Girard, ``Approximately bisimilar abstractions of incrementally stable finite or infinite dimensional systems,'' in {\em 53rd IEEE conference on decision and control}, pp.~824--829, 2014.

\bibitem{jagtap2020symbolic}
P.~Jagtap and M.~Zamani, ``Symbolic models for retarded jump--diffusion systems,'' {\em Automatica}, vol.~111, p.~108666, 2020.

\bibitem{angeli}
D.~Angeli, ``A {L}yapunov approach to incremental stability properties,'' {\em IEEE Transactions on Automatic Control}, vol.~47, no.~3, pp.~410--421, 2002.

\bibitem{zamanicharacterize}
M.~Zamani, N.~{van de Wouw}, and R.~Majumdar, ``Backstepping controller synthesis and characterizations of incremental stability,'' {\em Systems \& Control Letters}, vol.~62, no.~10, pp.~949--962, 2013.

\bibitem{pushpakHamilton}
P.~Jagtap and M.~Zamani, ``Backstepping design for incremental stability of stochastic {H}amiltonian systems with jumps,'' {\em IEEE Transactions on Automatic Control}, vol.~63, no.~1, pp.~255--261, 2018.

\bibitem{zamaninonsmooth}
M.~Zamani and N.~van~de Wouw, ``Controller synthesis for incremental stability: Application to symbolic controller synthesis,'' in {\em European Control Conference (ECC)}, pp.~2198--2203, 2013.

\bibitem{backsteppingzamani}
M.~Zamani and P.~Tabuada, ``Backstepping design for incremental stability,'' {\em IEEE Transactions on Automatic Control}, vol.~56, no.~9, pp.~2184--2189, 2011.

\bibitem{GPBook}
C.~K.~I. Williams and C.~E. Rasmussen, {\em Gaussian processes for machine learning}, vol.~2.
\newblock MIT press Cambridge, MA, 2006.

\bibitem{tracking}
T.~Beckers, D.~Kuli\'{c}, and S.~Hirche, ``Stable {G}aussian process based tracking control of {E}uler–{L}agrange systems,'' {\em Automatica}, vol.~103, p.~390–397, may 2019.

\bibitem{feedbackLinearization}
J.~Umlauft and S.~Hirche, ``Feedback linearization based on {G}aussian processes with event-triggered online learning,'' {\em IEEE Transactions on Automatic Control}, vol.~65, no.~10, pp.~4154--4169, 2020.

\bibitem{hircheControl}
J.~Umlauft, L.~Pöhler, and S.~Hirche, ``An uncertainty-based control {L}yapunov approach for control-affine systems modeled by {G}aussian process,'' {\em IEEE Control Systems Letters}, vol.~2, no.~3, pp.~483--488, 2018.

\bibitem{JagtapGP}
P.~Jagtap, G.~J. Pappas, and M.~Zamani, ``Control barrier functions for unknown nonlinear systems using {G}aussian processes,'' in {\em 59th IEEE Conference on Decision and Control (CDC)}, pp.~3699--3704, 2020.

\bibitem{ISpS}
A.~Mironchenko, ``Criteria for input-to-state practical stability,'' {\em IEEE Transactions on Automatic Control}, vol.~64, no.~1, pp.~298--304, 2019.

\bibitem{kolathaya2018input}
S.~Kolathaya and A.~D. Ames, ``Input-to-state safety with control barrier functions,'' {\em IEEE control systems letters}, vol.~3, no.~1, pp.~108--113, 2018.

\bibitem{jagtap2020control}
P.~Jagtap, G.~J. Pappas, and M.~Zamani, ``Control barrier functions for unknown nonlinear systems using {G}aussian processes,'' in {\em 59th IEEE Conference on Decision and Control (CDC)}, pp.~3699--3704, 2020.

\bibitem{HircheGP}
J.~Umlauft, L.~Pöhler, and S.~Hirche, ``An uncertainty-based control {L}yapunov approach for control-affine systems modeled by {G}aussian process,'' {\em IEEE Control Systems Letters}, vol.~2, no.~3, pp.~483--488, 2018.

\bibitem{InfoGain}
N.~Srinivas, A.~Krause, S.~M. Kakade, and M.~W. Seeger, ``Information-theoretic regret bounds for {G}aussian process optimization in the bandit setting,'' {\em IEEE Transactions on Information Theory}, vol.~58, no.~5, pp.~3250--3265, 2012.

\bibitem{Khalil:1173048}
H.~K. Khalil, {\em {Nonlinear systems; 3rd ed.}}
\newblock Upper Saddle River, NJ: Prentice-Hall, 2002.

\bibitem{liu2019compositional}
S.~Liu and M.~Zamani, ``Compositional synthesis of almost maximally permissible safety controllers,'' in {\em American Control Conference (ACC)}, pp.~1678--1683, 2019.

\bibitem{deltaISS}
M.~Zamani and R.~Majumdar, ``A {L}yapunov approach in incremental stability,'' in {\em 50th IEEE Conference on Decision and Control and European Control Conference}, pp.~302--307, 2011.

\bibitem{2Regbook}
R.~M. Murray, S.~S. Sastry, and L.~Zexiang, {\em A Mathematical Introduction to Robotic Manipulation}.
\newblock USA: CRC Press, Inc., 1st~ed., 1994.

\end{thebibliography}

\end{document}